%% file: JSAIT_Running_File.tex
\documentclass[journal]{IEEEtran}
\usepackage{amssymb}
\usepackage{amsmath}
\usepackage{cite}
\usepackage{mathtools}
\usepackage{amsthm}
\usepackage{url}           
\usepackage[table]{xcolor}
\usepackage{setspace}
\usepackage{graphicx}
\graphicspath{ {./tikz/} }
\usepackage{tikz}
\usepackage{float}
\usepackage{tabularx}
\usepackage{algorithm2e}

\newtheorem{theorem}{Theorem}

\newtheorem{lemma}{Lemma}

\newtheorem{proposition}{Proposition}

\theoremstyle{definition}
\newtheorem{example}{Example}
\newtheorem{definition}{Definition}
\newtheorem{corollary}{Corollary}

\newcommand{\dist}{\operatorname{dist}}

\title{Optimal Binary Differential Privacy via Graphs}

\author{Sahel Torkamani, Javad B. Ebrahimi, Parastoo Sadeghi, Rafael G. L. D'Oliveira, and Muriel Médard
\thanks{Sahel Torkamani is with the Department of Mathematical Sciences at Sharif University of Technology, Tehran, Iran, Email: \texttt{sahel.torkamani@sharif.edu}. Javad Ebrahimi is with the Department of Mathematical Sciences at Sharif University of Technology, Tehran, Iran, Email: \texttt{javad.ebrahimi@sharif.edu}. Parastoo Sadeghi is with the School of Engineering and Technology, University of New South Wales, Canberra, Australia, Email: \texttt{p.sadeghi@unsw.edu.au}. Rafael G. L. D'oliveira is with the School of Mathematical and Statistical Sciences at Clemson University, Clemson, SC, Email: \texttt{rdolive@clemson.edu}. Muriel M\'edard is with the Research Laboratory of Electronics at the Massachusetts Institute of Technology, Cambridge, MA, Email \texttt{medard@mit.edu}. Preliminary results of this paper were partly presented in~\cite{RafnoDP2colorPaper}.}}

\begin{document}
\maketitle
\RestyleAlgo{ruled}

\begin{abstract}
We present the notion of \emph{reasonable utility} for binary mechanisms, which applies to all utility functions in the literature. This notion induces a partial ordering on the performance of all binary differentially private (DP) mechanisms. DP mechanisms that are maximal elements of this ordering are optimal DP mechanisms for every reasonable utility. By looking at differential privacy as a randomized graph coloring, we characterize these optimal DP in terms of their behavior on a certain subset of the boundary datasets we call a boundary hitting set. In the process of establishing our results, we also introduce a useful notion that generalizes DP conditions for binary-valued queries, which we coin as suitable pairs. Suitable pairs abstract away the algebraic roles of $\varepsilon,\delta$ in the DP framework, making the derivations and understanding of our proofs simpler. Additionally, the notion of a suitable pair can potentially capture privacy conditions in frameworks other than DP and may be of independent interest. 
\end{abstract}

\section{Introduction} 
Differential privacy (DP) \cite{dwork2006calibrating} has emerged as a leading standard in private data analysis \cite{survey_2017}. This framework has been instrumental in protecting privacy across a multitude of applications. Most prominently, the United States Census Bureau integrated differential privacy into its 2020 Census release \cite{USCBadopts}. Furthermore, industry leaders like Google~\cite{google}, Microsoft~\cite{microsoft}, and Apple~\cite{apple} have also incorporated DP into their respective systems. DP is also heavily studied and used in deep learning~\cite{abadi2016deep}, \cite{GANobfuscator} and federated learning \cite{advances:FL}, \cite{Wei20}.

Differential privacy is often achieved through a randomized perturbation of the true query outputs before sharing them with potentially untrustworthy entities. However, such a perturbation (also known as a DP mechanism) inevitably affects the reliability and utility of the output. Therefore, one of the central and challenging research problems in the field is how to design and implement DP mechanisms to best balance privacy and utility~\cite{Censusissues}. Under certain parameter settings and assumptions, optimal DP mechanisms have been studied in the literature for real-valued queries~\cite{natashalaplace,ghosh2012universally,staircase,quan2016optimal,geng2015optimal} and for categorical or binary-valued data \cite{holahan_LP_DP,multiparty,hamming_DP}.

Previous works on differential privacy have considered different utility functions to measure performance. Thus, while a certain DP mechanism might perform well, or even optimally, for a certain utility, it might not do so for another. In Definition~\ref{def: reasonable utility}, we present the notion of \emph{reasonable utility} for binary mechanisms, which applies to all binary utility functions in the literature. This notion induces a partial ordering on the performance of all binary DP mechanisms. DP mechanisms which are maximal elements of this ordering are optimal DP mechanisms for every reasonable utility. In Theorem~\ref{theo: uniqueness}, we characterize these optimal DP mechanisms. To do so, we look at differential privacy as a randomized graph coloring.

In our graph formulation, each vertex $v \in V$ of the graph represents a dataset and each edge represents a neighborhood relation. The true value of the binary query is represented by the vertex color (such as \texttt{blue} and \texttt{red}). A DP mechanism is then a randomized coloring of the vertices subject to local privacy constraints. We categorize datasets into boundary and non-boundary datasets. Boundary datasets are those with at least one neighbor with a different true query value (color), and non-boundary datasets are those in which no single individual in the dataset can change the query. 

Theorem~\ref{theo: uniqueness} shows that optimal DP mechanisms are characterized by the values of the DP mechanism on a certain subset of the boundary datasets we call a \emph{boundary hitting set}. Thus, if the values of a DP mechanism on a boundary hitting set are defined and satisfy DP conditions among themselves, then there exists a unique optimal DP mechanism which outperforms all others, for any reasonable utility function.

In the process of establishing our results, we also introduce a useful notion that generalizes DP conditions for binary-valued queries. We coin this as a \emph{suitable pair}, which abstracts away the algebraic roles of $\varepsilon,\delta$ in the DP framework and instead focuses on the following: a randomized binary mechanism defined on a dataset $v$ imposes an upper bound and a lower bound on the mechanism on a neighboring dataset $u$. These bounds at $u$, in turn, impose upper and lower bounds on the mechanism in the original dataset $v$. The strength of the notion of suitable pair is that non-local privacy conditions between non-neighboring datasets can be easily understood and manipulated without being entangled in algebraic DP conditions. Thus, simplifying the derivations and understanding of our proofs. Additionally, the notion of a suitable pair can potentially capture privacy conditions in frameworks other than DP and may be of independent interest.

\subsection{Main Contributions}

Our main contributions are as follows.

\begin{itemize}

    \item In Definition~\ref{def: reasonable utility}, we present the notion of \emph{reasonable utility} for binary mechanisms, which applies to all utility functions in the literature. This notion induces a partial ordering on the performance of all binary DP mechanisms. DP mechanisms which are maximal elements of this ordering are optimal DP mechanisms for every reasonable utility.

    \item In Theorem~\ref{theo: uniqueness} we characterize optimal DP mechanisms by their values on a certain subset of the boundary datasets, which we call a \emph{boundary hitting set}.

    \item In Definition~\ref{def: suitable pair} we present the notion of a \emph{suitable pair}. This notion generalizes DP conditions for binary-valued queries and abstracts away the algebraic roles of $\varepsilon,\delta$, thus simplifying our proofs. 

    \item We present Algorithm~\ref{algo: main}, for finding optimal mechanisms within the suitable pair framework, as well as a more efficient Algorithm~\ref{algo: single} for the case where one is solely interested in the output of a mechanism on a specific dataset. The optimality of Algorithm~\ref{algo: main} is stated in Theorem~\ref{theo: algo 1}. 
\end{itemize}

 Theorem~\ref{theo: uniqueness} generalizes the results in~\cite{RafnoDP2colorPaper}, which is stated as Corollaries \ref{theo: boundary homo} and~\ref{cor:balanced} for the spcial case of boundary homogenous mechanisms and balanced mechanisms, respectively. Definition~\ref{def: suitable pair}, Algorithms~\ref{algo: main} and~\ref{algo: single}, Theorem~\ref{theo: algo 1}, and associated intermediate results are all new in this paper with respect to~\cite{RafnoDP2colorPaper}.

\subsection{Paper Organization}

Section~\ref{sec:main} contains a statement of the problem and all main results of the paper. In Section~\ref{sec:setting}, we review basic DP definitions and introduce the notion of reasonable utility, mechanism utility dominance, and optimal mechanism. Section~\ref{sec: DP Graph} presents the DP mechanism as the randomized coloring of datasets on the graph. Section~\ref{sec:optimal}, highlighted in Theorem~\ref{theo: uniqueness}, presents the main results for optimally extending the  mechanism in terms of the restricted $(\varepsilon, \delta)$-DP mechanism on a boundary hitting set. Section~\ref{sec:suitable} generalizes the results of Section~\ref{sec:optimal} using the new notion of a suitable pair and presents the necessary and sufficient condition for the existence of the unique optimal extension of a mechanism restricted to a boundary hitting set. It also summarizes the optimal extension in Algorithm~\ref{algo: main}. All proofs are in Section~\ref{sec:proofs}.

\section{Main Results}\label{sec:main}

\subsection{Differential Privacy}\label{sec:setting}

We denote by $V$ the family of datasets. We consider a symmetric neighborhood relationship $\sim$ on $V$ where $u,v \in V$ are said to be neighbors if $u \sim v$. We also consider a finite output space $Q$, which corresponds to the space over which the output of the queries lies. A randomized mechanism, which we refer to as just a mechanism, is a random function $\mathcal{M}:V \to Q$, from the family of datasets to the output space.

\begin{definition}[Differential Privacy \cite{dwork2014algorithmic}]\label{def: dp} Let $\varepsilon, \delta \in \mathbb{R}$ be such that $\varepsilon \geq 0$ and $0\leq \delta<1$. Let $V$ be a set and $\sim$ be a symmetric relation on $V$. Then, a mechanism $\mathcal{M}:V \to Q$ is $(\varepsilon,\delta)$-differentially private if for any $u \sim v$ and $S \subseteq Q$, we have  $\Pr [\mathcal{M}(u) \in \mathcal{S}] \leq e^\varepsilon \Pr [\mathcal{M}(v) \in \mathcal{S}] + \delta $. We denote the set of all $(\varepsilon, \delta)$-DP mechanisms $\mathcal{M}:V \rightarrow Q$ by $\mathfrak{M}_{\varepsilon,\delta}(V, Q)$. However, when $V$ and $Q$ are clear from the context, we refer to $\mathfrak{M}_{\varepsilon,\delta}(V, Q)$ as $\mathfrak{M}_{\varepsilon,\delta}$.
\end{definition}

In this paper, we consider the case where the size of the output space is $|Q|=2$, i.e., binary-valued queries. Without loss of generality, we set the output space to $Q=\{\texttt{blue},\texttt{red}\}$. The DP conditions for any $u\sim v$ in $V$ are then as follows.
\begin{align} \label{eq: p1}
	\Pr[\mathcal{M}(u) = \texttt{blue}] &\leq e^\varepsilon  \Pr[\mathcal{M}(v) = \texttt{blue}]+ \delta, \\
	1-\Pr[\mathcal{M}(u) = \texttt{blue}] &\leq e^\varepsilon (1-\Pr[\mathcal{M}(v) = \texttt{blue}]) + \delta,\label{eq: p2}\\
\Pr[\mathcal{M}(v) = \texttt{blue}] &\leq e^\varepsilon  \Pr[\mathcal{M}(u) = \texttt{blue}]+ \delta, \label{eq: p3}\\
	1-\Pr[\mathcal{M}(v) = \texttt{blue}] &\leq e^\varepsilon (1-\Pr[\mathcal{M}(u) = \texttt{blue}]) + \delta.\label{eq: p4}
	\end{align}
Since we only consider binary-valued queries, we have that $\Pr[\mathcal{M}(u) = \texttt{red}] = 1-\Pr[\mathcal{M}(u) = \texttt{blue}]$ and that $\Pr[\mathcal{M}(v) = \texttt{red}]= 1- \Pr[\mathcal{M}(v) = \texttt{blue}]$.

We consider a function $T: V \to Q$, which we refer to as the true function. Our goal is to approximate the true function $T$ using an $(\varepsilon,\delta)$-differentially private mechanism $\mathcal{M}$. To measure the performance of the mechanism, i.e., how good the approximation is, a utility function $\mathcal{U}_T:\mathfrak{M}_{\varepsilon,\delta} \rightarrow \mathbb{R}$ must be defined, where $\mathcal{U}_T[\mathcal M] \geq \mathcal{U}_T[\mathcal M']$ means that the mechanism $\mathcal M$ outperforms $\mathcal M'$ with respect to the true function $T$. In this work, we do not consider a specific utility function, but rather consider a general family of them.

\begin{definition} \label{def: reasonable utility}
A utility function $\mathcal{U}:\mathfrak{M}_{\varepsilon,\delta} \rightarrow \mathbb{R}$ is \textit{reasonable} if $\Pr[\mathcal{M}(u)=T(u)] \geq \Pr[\mathcal{M'}(u)=T(u)]$ for every $u \in V$ implies $\mathcal{U}[\mathcal{M}] \geq \mathcal{U}[\mathcal{M}']$. When this condition holds, we say that the mechanism $\mathcal{M}$ dominates $\mathcal{M}'$.
\end{definition}

Given the true function $T$, the notion of domination in Definition~\ref{def: reasonable utility} induces a partial order on the set $\mathfrak{M}_{\varepsilon,\delta}$ of all $(\varepsilon,\delta)$-DP mechanisms. If a mechanism $\mathcal M$ dominates another mechansim $\mathcal M'$ then the first one outperforms the second for every reasonable utility function. It is not always the case that two mechanisms can be compared, even when restricted to a reasonable utility. We give an example below.

\begin{example}
Consider the dataset $V=\{1,2\}$ where $1 \sim 2$ and the true function $T:V \rightarrow Q$ is such that $T(1) = \texttt{blue}$ and $T(2)=\texttt{red}$. Let $\mathcal{M}_1$ and $\mathcal{M}_2$ be two $(\log(2),0.1)$-DP mechanisms\footnote{In this paper, by $\log$ we mean the natural logarithm.} defined such that $\Pr[\mathcal{M}_1(1)=\texttt{blue}] = 0.58$, $\Pr[\mathcal{M}_1(2)=\texttt{red}] = 0.76$, $\Pr[\mathcal{M}_2(1)=\texttt{blue}] = 0.64$, and $\Pr[\mathcal{M}_2(2)=\texttt{red}] = 0.73$. Then, neither mechanism dominates the other. The reason for this is that there are reasonable utility functions which, for a mechanism $\mathcal{M} \in \mathfrak{M}_{\varepsilon,\delta}$ might prefer a higher value for $\Pr[\mathcal{M}(1)=\texttt{blue}]$ more than 
a higher value for $\Pr[\mathcal{M}(2)=\texttt{red}]$, or vice-versa. Extreme cases of this are the reasonable utility functions $U[\mathcal{M}] = \Pr[\mathcal{M}(1)=\texttt{blue}]$ and $U'[\mathcal{M}] = \Pr[\mathcal{M}(2)=\texttt{red}]$, both disagreeing on which of $\mathcal{M}_1$ or $\mathcal{M}_2$ is better.
\end{example}

For more discussion and insight on the notion of reasonable utility and its extensions, see Sections~\ref{sec:discussion1} and~\ref{sec:discussion2}. 

We are interested in characterizing the optimal $(\varepsilon, \delta)$-DP mechanisms, i.e., the $(\varepsilon, \delta)$-DP mechanisms $\mathcal{M}$ which are not dominated by any other mechanism. These correspond to the maximal elements in the partial order $\mathfrak{M}_{\varepsilon,\delta}$. To find such mechanisms, we reinterpret the problem as a randomized graph coloring problem, which we describe in Section \ref{sec: DP Graph}.

This random graph coloring approach, together with the notion of suitable pairs defined in Section~\ref{sec:suitable}, allows us to abstract the problem of characterizing optimal $(\varepsilon, \delta)$-DP mechanisms. Through such abstraction, we show in Theorem~\ref{theo: uniqueness} that if the mechanism is defined only on an appropriate subset of neighboring datasets, between which the true function changes value, then the optimal mechanism can be uniquely found for every other dataset.

\subsection{Differential Privacy as Randomized Graph Colorings} \label{sec: DP Graph}

We interpret differential privacy as a randomized graph coloring problem. The vertices of the graph\footnote{We assume the graph is undirected and connected. Otherwise, the results of this paper apply to any connected component of $\mathcal{G}(V,E)$.} $\mathcal{G}(V,E)$ are the datasets $u \in V$ and the edges $E$ are the neighboring relation on the datasets, i.e. two vertices $u,v \in V$ have an edge between them if $u \sim v$. The true function $T: V \rightarrow Q$ is a graph coloring ($Q$ is the set of colors) of the vertices of $\mathcal{G}$. For a given color $j \in Q$, the inverse image $T^{-1}(j)$ is the set of vertices with true value $j$. Therefore, for a given vertex $u\in V$, the inverse image $T^{-1}(T(u))$ is the set of vertices with the same true value as $u$.  Since differential privacy is a local condition, i.e., it is a condition on the $u,v \in V$ such that $u \sim v$, the notion of a neighborhood is essential.

\begin{figure}
\begin{center}
\input{./tikz/randomGraph.tikz}
\end{center}

\caption{In this graph, the vertices are the datasets and the edges are the neighborhood relationships, e.g., since there exists an edge between $a$ and $b$, it follows that $a \sim b$. The true function is a graph coloring of the vertices, e.g., $T(a)=\texttt{blue}$ and $T(d)=\texttt{red}$. The set of vertices with true value $\texttt{blue}$ is $T^{-1}(\texttt{blue})=\{a,b,g,h,k,\ell,n,s,t,u,v\}$ and the ones with true value \texttt{red} is $T^{-1}(\texttt{red})= \{d,e,f,i,j,m,n,o,p,q,r\}$. The neighbohood of the vertex set $\{g,h\}$ is $N(\{g,h\}) = \{a,b,c,i,n\}$. The distance between $a$ and $b$ is $\dist(a,d) = 5$, which is the size of the shortest path $(a,g),(g,c),(c,h),(h,i),(i,d)$ between them. The edges colored in green are the boundary edges $\partial _T (\mathcal{G}, E) = \{ (\ell , q), (h,n), (h,i), (r,s), (r,t), (f,u) \}$. The boundary vertices are the vertices $\partial _T (\mathcal{G}, V) = \{ \ell,q,h,n,i,r,s,t,f,u \}$. The sets $\mathcal{H}_T = \{ l,h,u,s,t\}$ and $\mathcal{H}'_T = \{a, q,n,i,s,t,u,v\}$ are boundary-hitting sets, because they contain at least one vertex from each boundary edge, while the set $\{ a, \ell , n,s,t,j,u\}$ is not a boundary-hitting set since it does not include any vertex from the boundary edge $(h,i)$. Finally, $\partial_T(\mathcal{G},\texttt{blue}) = \{h,\ell,s,t,u\}$ and $\partial_T(\mathcal{G},\texttt{red}) = \{f,i,n,q,r\}$. }\label{fig: randomGraphColoring}
\end{figure}
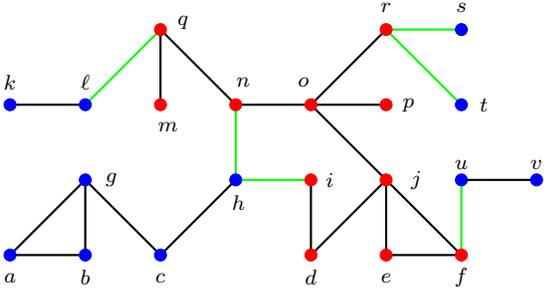

\begin{definition}[Neighborhood]
The neighborhood of a subset $S \subseteq V$ of vertices, denoted by $N(S)$, is the set of all vertices in $V - S$ which are neighbors to at least one element of $S$.
\end{definition}

\begin{definition}\label{def:path}
 A sequence of vertices $u = u_0, u_1, \cdots, u_n = v$ is said to form a path from $u$ to $v$, denoted by $(u,v)$-path, if $(u_0, u_1), \cdots, (u_{n-1},u_n) \in E$. We say $n$ is the path length. The distance between two nodes $u,v$, denoted by $\dist(u,v)$, is the shortest path length from $u$ to $v$. The distance between two subsets $A_1, A_2 \subset V$ is the shortest path length between any $a_1\in A_1$ and any $a_2 \in A_2$.
\end{definition}

\begin{definition}[Boundary Edges]\label{def: boundary}
The \emph{boundary edge} of $\mathcal G$ with respect to $T$, denoted by $\partial _T (\mathcal{G}, E)$, is the set of the edges in $\mathcal G$ whose two endpoints have different true query values.
\end{definition}

\begin{definition}[Boundary Vertices]
    The \emph{boundary vertices} of $\mathcal G$ with respect to $T$, denoted by $\partial _T (\mathcal{G}, V)$, is the set of all the endpoints of the boundary edges.    
\end{definition}

\begin{definition}[Boundary-hitting Set]\label{def: boundary-hitting set}
A \emph{boundary-hitting set} of $\mathcal G$ with respect to $T$, denoted by $\mathcal{H}_T$, is a subset of the vertices which contains at least one endpoint of every edge in $\partial _T (\mathcal{G}, E)$.
\end{definition}

\begin{definition}[Boundary Vertices with True Value $j$]
     The boundary vertices with true value $j$ is the set of boundary vertices whose true value is $j$. We denote this set by $\partial_T(\mathcal{G},j) := \partial_T(\mathcal{G},V) \cap T^{-1}(j)$.   
\end{definition}

\noindent In Fig.~\ref{fig: randomGraphColoring} we illustrate these definitions.

In an $(\varepsilon,\delta)$-differentially private mechanism, every path connecting two vertices $u,v \in V$ induces an upper bound on the probability of the mechanism output, i.e., the probability $\Pr[\mathcal{M}(v)=j]$ induces an upper bound on $\Pr[\mathcal{M}(u)=j]$ and vice versa. These upper bounds are induced by the $(\varepsilon, \delta)$-DP conditions (inequalities \eqref{eq: p1} through  \eqref{eq: p4}). 
Since each upper bound that $v$ induces on $u$ is an increasing function of $\alpha := \Pr[\mathcal{M}(v)=j]$ and of the length of the path between $u$ and $v$, then the shortest path induces the tightest upper bound on $\Pr[\mathcal{M}(u)=j]$. We formalize this statement through Definition \ref{def:UB} and Proposition \ref{prop: single}.

\begin{definition}\label{def:UB}
Let $u, v$ be distinct vertices in $V$ with distance $d = \dist(u,v)$ between them and $\alpha \in [0, 1]$ be some fixed value. Then, the probability induced on the vertex $u$ by the vertex $v$ with value $\alpha$ is:
\begin{align}\label{eq:UB}
    p(d,\alpha) = \begin{cases}
    e^{d \varepsilon } \alpha +
    \delta \frac{ e^{d \varepsilon }-1}{e^{\varepsilon}-1}, &    d \leq \tau,  \\ 
    \min \big(1, e^{(2\tau - d)\varepsilon} \alpha + 1 - \frac{1}{e^{(d - \tau)\varepsilon}}
    &  \tau < d. \\
       \quad + \frac{\delta (
   e^{\tau \varepsilon}+e^{(d - \tau) \varepsilon}-2)}{e^{(d - \tau) \varepsilon}(e^{\varepsilon}-1)}\big) ,
    \end{cases}
\end{align}
where 
\begin{align} \label{eq:tau0}
 \tau = \left \lceil \frac{1}{\varepsilon} \log{\frac{(e^\varepsilon + 2\delta -1) }{(e^\varepsilon +1)(e^{\varepsilon}\alpha - \alpha + \delta )}} \right \rceil.
 \end{align}
\end{definition}

Taking the minimum in the second line of \eqref{eq:UB} ensures that $p(d,\alpha)$ never exceeds $1$. The next proposition establishes how each vertex $v$ induces an upper bound on $\Pr[\mathcal{M}(u)=j]$ for any other vertex $u$.

\begin{proposition}\label{prop: single}
Let $\mathcal{M}:V \rightarrow Q$ be a mechanism. Then, $\mathcal{M}$ is an $(\varepsilon,\delta)$-DP mechanism if and only if for every $j \in \{\texttt{blue},\texttt{red}\}$ and every distinct vertices $u, v \in V$ with $d = \dist(u,v)$ and $\Pr[\mathcal{M}(u)=j] = \alpha$, it holds that $\Pr[\mathcal{M}(v)=j] \leq p(d,\alpha)$.
\end{proposition}

Proposition~\ref{prop: single} presents a closed form expression for the differential privacy condition on non-neighboring datasets, which we use, in Theorem~\ref{theo: uniqueness}, to characterize optimal mechanisms.

\subsection{Characterizing Optimal Mechanisms}\label{sec:optimal}

We show that optimal mechanisms are uniquely characterized by their behavior on the boundary edges, i.e., edges connecting vertices with different true values (see Definition~\ref{def: boundary}). We do this by showing that when a mechanism has been predefined on a boundary-hitting set (see Definition~\ref{def: boundary-hitting set}), then it can be uniquely extended to an optimal mechanism over all other vertices. These results are shown in Theorem~\ref{theo: uniqueness} and Algorithm~\ref{algo: main}. To this end, we introduce the notions of mechanism restriction and extension.

\begin{definition}[Mechanism Restriction] The restriction of a mechanism $\mathcal{M} : \mathcal{V} \rightarrow Q$ to a subset $A\subseteq V$ is  $\mathcal{M}|_A:A\rightarrow Q$.
\end{definition}
\noindent We also refer to $\mathcal{M}|_A$ as a partial mechanism.

\begin{definition}[Mechanism Extension]\label{def:defintion}
Let $A \subseteq V$ and $\mathcal{M}':A \rightarrow Q$ be a mechanism. Then, a mechanism $\mathcal{M}:V \rightarrow Q$ is an extension of $\mathcal{M}'$ if $\mathcal{M}|_A = \mathcal{M}'$.
\end{definition}

\begin{theorem} \label{theo: uniqueness}
Let $\mathcal{H}_T$ be a boundary-hitting set with respect to the true function $T$ and $\mathcal{M}': \mathcal{H}_T \rightarrow Q$  be a randomized function. Then, there exists a unique optimal $(\varepsilon,\delta)$-DP mechanism $\mathcal{M}:V \rightarrow Q$ such that $\mathcal{M}|_{\mathcal{H}_T} = \mathcal{M}'$ if and only if for every $u,v \in \mathcal{H}_T$ with $d = \dist(u,v)$, we have $\Pr[\mathcal{M}'(u) = j] \leq p(d,\Pr[\mathcal{M}'(v)=j])$ for a fixed $j \in Q$. Moreover, for every $u \notin \mathcal{H}_T$, the optimal mechanism is \begin{align}\label{eq: unique1}   \Pr[\mathcal{M}(u) = T(u)] = \min \limits_{v \in \mathcal{H}_T} p(\dist(u,v),\Pr[\mathcal{M}(v) = T(u)]).
\end{align}
\end{theorem}
Theorem~\ref{theo: uniqueness} establishes that optimal mechanisms are uniquely characterized by their values at the boundary. Moreover, the assumption that $\mathcal{H}_T$ is a boundary-hitting set is essential to the theorem, as the following example shows.

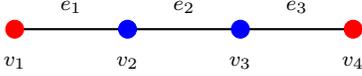
\begin{figure}[t]
\begin{center}
  \input{./tikz/counter.tikz}  
\end{center}
\caption{
The graph $\mathcal{G}$ in Examples~\ref{ex: counter} and~\ref{ex: algorithm}. In Example~\ref{ex: counter}, we show that if the set on which the restricted mechanism is defined is not a boundary hitting set, then there might not be a unique optimal extension for the setting of Theorem~\ref{theo: uniqueness}. In Example ~\ref{ex: algorithm}, we show how Algorithm~\ref{algo: main} works when the restricted mechanism is properly defined on the boundary hitting set $\mathcal{H}_T = \{v_1, v_4\}$.}\label{fig: counter}
\end{figure} 

\begin{example}\label{ex: counter}

Let $\mathcal{G}(V,E)$ be the graph with vertices $V = \{v_1, \ldots, v_4\}$ and edges $E=\{(v_1,v_2),(v_2,v_3),(v_3,v_4)\}$, the true function be such that $T(v_1)=T(v_4) = \texttt{red}$ and $T(v_2)=T(v_3) = \texttt{blue}$, and $A=\{v_3\}$. We illustrate this in Fig.~\ref{fig: counter}.

Note that $A$  is not a boundary-hitting set, since it is missing a vertex from $(v_1,v_2)$. We now show a $(\log(2),0)$-DP mechanism $\mathcal{M}': A \rightarrow Q$ which does not have a unique optimal $(\log(2),0)$-DP extension to all of $V$.

Let $\mathcal{M}'$ be such that $\Pr[\mathcal{M}'(v_3)=\texttt{blue}] = \frac{1}{2}$. Then, the $(\log(2),0)$-DP mechanisms $\mathcal{M}_1: V \rightarrow Q$ such that 
\begin{align*}
    \Pr[\mathcal{M}_1(v_1) &= \texttt{blue}]= \frac{1}{2}, \\[5pt]
    \Pr[\mathcal{M}_1(v_2) &= \texttt{blue}]= \frac{3}{4}, \\[5pt]
    \Pr[\mathcal{M}_1(v_3) &= \texttt{blue}]= \frac{1}{2}, \\[5pt]
    \Pr[\mathcal{M}_1(v_4) &= \texttt{blue}]= \frac{1}{4},
\end{align*}
and $\mathcal{M}_2: V \rightarrow Q$ such that
\begin{align*}
    \Pr[\mathcal{M}_1(v_1) &= \texttt{blue}]= \frac{1}{8}, \\[5pt]
    \Pr[\mathcal{M}_1(v_2) &= \texttt{blue}]= \frac{1}{4}, \\[5pt]
    \Pr[\mathcal{M}_1(v_3) &= \texttt{blue}]= \frac{1}{2}, \\[5pt]
    \Pr[\mathcal{M}_1(v_4) &= \texttt{blue}]= \frac{1}{4},
\end{align*}
are extensions of $\mathcal{M}'$ which are not comparable. Indeed, each is a maximal element in the partially ordered set of $(\log(2),0)$-DP mechanisms on $\mathcal{G}$.
\end{example}

Theorem \ref{theo: uniqueness} generalizes the main results of \cite{RafnoDP2colorPaper}, which we restate as corollaries below.

\begin{corollary}\label{theo: boundary homo}
Let $\alpha_\texttt{blue},\alpha_\texttt{red} \in [0,1]$ be fixed real numbers. Suppose there exists an optimal $(\varepsilon,\delta)$-DP mechanism $\mathcal{M}:V \rightarrow Q$ satisfying $\Pr[\mathcal{M}(v)=T(v)] = \alpha_{T(v)}$, for every boundary vertex $v \in \partial _T (\mathcal{G}, V)$. Then for every vertex $u \notin \partial _T (\mathcal{G}, V)$, the optimal mechanism must satisfy $\Pr[\mathcal{M}(u) = T(u)] = p(\dist(u,w),\alpha_{T(w)})$, where $w \in \partial_T(\mathcal{G},T(u))$ is the closest boundary vertex to $u$.\footnote{Due to a different labelling of vertices in \cite{RafnoDP2colorPaper}, $\tau$ in \eqref{eq:tau0} is larger than the corresponding $\tau$ in \cite[Definition 8]{RafnoDP2colorPaper} by one. After appropriate transformations, they both result in the same expression for $\Pr[\mathcal{M}(u)=T(u)]$.}
\end{corollary}

Corollary \ref{theo: boundary homo} states that when the restricted mechanism is homogeneous on the boundary, i.e., $\Pr[\mathcal{M}'(u) = T(u)] = \Pr[\mathcal{M}'(w) = T(w)]$ for every $v,w \in \partial_T(\mathcal{G},V)$ such that $T(v)=T(w)$, finding the optimal extension to a non-boundary vertex $u$ using \eqref{eq: unique1} reduces to first finding the closest boundary vertex to $u$, denoted by $w$. Note that by definition of the boundary and the distance, $w$ must have the same true value as $u$. Then, $\Pr[\mathcal{M}(u) = T(u)]$ will be given by \eqref{eq:UB}  with $\alpha = \Pr[\mathcal{M}(w)=T(w)]$ and $d = \dist(u,w)$. A particularly interesting boundary homogeneous case is when the $(\varepsilon,\delta)$-DP mechanism is balanced, i.e., when $\alpha_\texttt{blue} = \alpha_\texttt{red}$, which is stated below.

\begin{corollary}\label{cor:balanced}
In the setting of Corollary~\ref{theo: boundary homo}, suppose $\alpha_\texttt{blue} = \alpha_\texttt{red}$. Then, there exists a unique optimal mechanism and it is such that for every $u \notin \partial _T (\mathcal{G}, V)$,
\begin{align*}
    \Pr[\mathcal{M}(u) = T(u)] = 1 - \frac{e^\varepsilon - 1 - \delta (e^{\varepsilon(d+1)} + e^{d \varepsilon} - 2)}{e^{d \varepsilon} (e^\varepsilon + 1) (e^{\varepsilon} -1)} ,
\end{align*}
where $d$ is the distance of $u$ to the boundary $\partial _T (\mathcal{G}, V)$.
\end{corollary}

\subsection{Suitable Pairs}\label{sec:suitable}
To prove our results, we introduce a generalized framework that captures the key conditions of differential privacy. We note that differential privacy imposes local constraints on neighboring vertices. Although non-neighboring vertices ultimately constrain each other, they do so only through intermediate neighboring vertices. These constraints are realized through upper and lower bounds on the probability of the mechanism outputting a value, as captured in \eqref{eq: p1}-\eqref{eq: p4}. Combining \eqref{eq: p2} and \eqref{eq: p3}, we obtain the upper bound 
\begin{align} \label{eq: upper dp}
  U_{\texttt{DP}}(\alpha)  := \min(e^{\varepsilon} \alpha + \delta , \frac{e^{\varepsilon}+\delta-1 +\alpha}{e^{\varepsilon}},1),
\end{align}
where $\alpha = \Pr[\mathcal{M}(u) = j]$. Analogously, combining \eqref{eq: p1} and \eqref{eq: p4}, we obtain the lower bound \begin{align} \label{eq: lower dp}
 L_{\texttt{DP}}(\alpha) := \max (e^{\varepsilon} \alpha- \delta - e^{\varepsilon} +1 , \frac{\alpha- \delta}{e^{\varepsilon}},0).
\end{align}

We generalize this notion in the following definition.

\begin{definition}[Suitable Pair] {\label{def: suitable pair}}
Let $L, U: [0,1] \to [0,1]$ be two increasing functions. We call $(L,U)$ a suitable pair if for every $\alpha \in [0,1]$ the following three properties hold.
\begin{enumerate}
    \item $L(\alpha) \leq \alpha \leq U(\alpha)$,
    \item $L(U(\alpha)) \leq \alpha \leq U(L(\alpha))$,
    \item $U(\alpha) \leq 1- L(1-\alpha)$.
\end{enumerate}
\end{definition}

The $(\varepsilon,\delta)$-DP lower and upper bounds in \eqref{eq: upper dp} and \eqref{eq: lower dp} are then a special case of $(L, U)$ suitable pair.

\begin{proposition}\label{prop: uniqueness}
The functions $U_{\texttt{DP}}$ and $L_{\texttt{DP}}$ are a suitable pair.
\end{proposition}

The notion of a suitable pair abstracts away the detailed algebraic expressions of differential privacy, e.g., those appearing in \eqref{eq: p1}-\eqref{eq: p4}. Specifically, the composition $U^d$ captures the upper bound condition that a vertex $u$ imposes on other vertices at distance $d$.\footnote{For $d \geq 1$, $U^d$ denotes $d$ compositions of the function $U$. For function $L$, $L^d$ is defined similarly.} We now define the notion of privacy in the suitable pair framework.

\begin{definition}[$(L,U)$-Privacy]\label{def:lu:private}
Let $\mathcal{G}(V,E)$ be a graph and $(L, U)$ be a suitable pair. We say that a randomized mechanism $\mathcal{M}:V \to Q$ is $(L,U)$-private if, for any $u \sim v$ and $j \in Q$, it holds that $\Pr[\mathcal{M}(v) = j] \in [L(\alpha),U(\alpha)]$, where $\alpha = \Pr[\mathcal{M}(u) = j]$.
\end{definition}

In Theorem~\ref{theo: algo 1} we generalize Theorem~\ref{theo: uniqueness} to suitable pairs. We begin by showing an intermediate lemma that specifies the necessary and sufficient conditions for the existence of a mechanism extension.

\begin{lemma} \label{lem: cons}
Let $(L, U)$ be a suitable pair and $\mathcal{M}': \mathcal{H}_T \rightarrow Q$ be a randomized function on a boundary-hitting set $\mathcal{H}_T$. Then, there exists an $(L, U)$-private mechanism $\mathcal{M}:V \to Q$, extending $\mathcal{M}'$, if and only if, for every $u,v \in \mathcal{H}_T$ and $d = \dist(u,v)$, we have $\alpha \leq U^d(\beta)$ and $\beta \leq  U^d(\alpha)$, where $\alpha = \Pr[\mathcal{M'}(u)= j]$ and $\beta = \Pr[\mathcal{M'}(v)= j]$ for an arbitrarily chosen $j$ from $Q$.
\end{lemma}

\begin{theorem} \label{theo: algo 1}
Let $(L, U)$ be a suitable pair and $\mathcal{M}': \mathcal{H}_T \rightarrow Q$ be a randomized function on a boundary-hitting set $\mathcal{H}_T$. Then, Algorithm \ref{algo: main} either outputs the optimal $(L,U)$-private extension of $\mathcal{M}'$ or no $(L,U)$-private extension $\mathcal{M}: V \rightarrow Q$ of $\mathcal{M}'$ exists.
\end{theorem}

\begin{algorithm}[h] 
\caption{The Optimal $(L,U)$-Private Extension}\label{algo: main}
    \SetAlgoLined
    \SetAlgoNoLine
    \SetAlgoLined
    \textbf{Input:} Graph $\mathcal{G}(V,E)$, true function $T$, suitable pair functions $(L,U)$,  boundary-hitting set $\mathcal{H}_T \subseteq V$,  randomized function $\mathcal{M}': \mathcal{H}_T \rightarrow Q$.
    
    \textbf{Output:} Optimal $(L,U)$-private extension $\mathcal{M}: V \rightarrow Q$ of $\mathcal{M}'$ if one exists.
        
    Choose $j\in Q$

    \For{$w, v \in \mathcal{H}_T$}{
    $d = \dist(w, v).$ 
    
    \If{
    $\Pr[\mathcal{M}'(w)=j] > U^{d}(\Pr[\mathcal{M}'(v)=j])$ 
    \\ \textit{\textbf{or}}
    $\Pr[\mathcal{M}'(v)=j] > U^{d}(\Pr[\mathcal{M}'(w)=j])$
    }
    {
    \Return {\textit{``No $(L,U)$-private  extension exists.''}}
    }}
    
    \For{$w \in V - \mathcal{H}_T$}
    {$\Pr[\mathcal{M}(w)=T(w)] = \min \limits_{u \in \mathcal{H}_T} U^{\dist(w,u)}(\Pr[\mathcal{M}'(u)=T(w)])$.}
    \Return Optimal extension $\mathcal{M}: V \rightarrow Q$ of $\mathcal{M}'$.
        
\end{algorithm}

Algorithm \ref{algo: main} works as follows. First, it checks whether the randomized function $\mathcal{M}'$ is extensible at all. If there exists a vertex $w \in \mathcal{H}_T$ such that the probability of outputting the true value $T(w)$ exceeds the $(L,U)$ bounds imposed by all other vertices in $\mathcal{H}_T$, then an extension is not possible.\footnote{In Algorithm \ref{algo: main}, we have fixed $j \in Q$ at the beginning of the algorithm. However, this is not necessary. Based on Lemma~\ref{lem:UBsufficient}, it is possible to select a different $j$ in each iteration of the for-loop.} Otherwise, for each $u$ not in $\mathcal{H}_T$ the algorithm assigns $\Pr[\mathcal{M}(u)=T(u)]$ to be the minimum upper bound imposed by the vertices in $\mathcal{H}_T$. In this way, it obtains the unique optimal $(L,U)$-private extension of $\mathcal{M}'$. Theorem~\ref{theo: uniqueness} follows from Algorithm~\ref{algo: main} by setting the upper and lower bound functions for DP according to \eqref{eq: upper dp} and \eqref{eq: lower dp}.

From a computational complexity point of view, the significance of Algorithm \ref{algo: main} is as follows. For a given vertex $u$, a naive approach would consider every possible path between $u$ and all other vertices in the graph to determine if it satisfies the privacy constraints. Whereas Algorithm \ref{algo: main} shows that as long as $\mathcal{H}_T$ is a boundary-hitting set, it is sufficient to consider paths between $u$ and $\mathcal{H}_T$, thus reducing computational complexity. Moreover, one does not need to consider all paths, but only the shortest path between $u$ and each $w \in \mathcal{H}_T$.

However, in many applications, one might not necessarily be interested in retrieving the whole optimal mechanism but instead in evaluating it on a particular dataset $u \in V$, reducing complexity even further. For this, we present Algorithm~\ref{algo: single}.

\begin{algorithm}[htb] 
\caption{The Optimal $(L,U)$-Private Extension Evaluated at on Particular Dataset}\label{algo: single}
    \SetAlgoLined
    \SetAlgoNoLine
    \SetAlgoLined
    \textbf{Input:} Graph $\mathcal{G}(V,E)$, true function $T$, suitable pair functions $(L,U)$,  boundary-hitting set $\mathcal{H}_T \subseteq V$,  randomized function $\mathcal{M}': \mathcal{H}_T \rightarrow Q$, a vertex $u \in V-\mathcal{H}_T$.
    
    \textbf{Output:} Optimal $(L,U)$-private extension $\mathcal{M}(u)$.
    
    Choose $j\in Q$
    
    \For{$w, v \in \mathcal{H}_T$}{
    $d = \dist(w, v).$ 
    
    \If{
    $\Pr[\mathcal{M}'(w)=j] > U^{d}(\Pr[\mathcal{M}'(v)=j])$ 
    \\ \textit{\textbf{or}}
    $\Pr[\mathcal{M}'(v)=j] > U^{d}(\Pr[\mathcal{M}'(w)=j])$
    }
    {
    \Return {\textit{``No $(L,U)$-private  extension exists.''}}
    }}
    
    \Return $ \Pr[\mathcal{M}(u)=T(u)] = \min \limits_{w \in \mathcal{H}_T} U^{\dist(w,u)}(\Pr[\mathcal{M}'(w)=T(u)])$.
        
\end{algorithm}

Whereas in Algorithm~\ref{algo: main} we must compute all shortest paths between vertices in $V-\mathcal{H}_T$ and those in $\mathcal{H}_T$, in Algorithm~\ref{algo: single} we need only to compute the shortest path between $u$ and $\mathcal{H}_T$. If we denote the ball centered at the vertex $u$ with radius at maximum distance to $\mathcal{H}_T$ by $\texttt{B}$, then the complexity of finding the shortest path between $u$ and $\mathcal{H}_T$ using the Dijkstra Algorithm \cite{cormen2022introduction} is $\Theta(|E(\texttt{B})|+|V(\texttt{B})|\log|V(\texttt{B})|)$, where $E(\texttt{B})$ and $V(\texttt{B})$ are the edges and vertices included in the ball, respectively. The complexity of checking for the existence of an extension is $\mathcal{O}(|\mathcal{H}_T|^2)$. Thus, the time complexity of Algorithm~\ref{algo: single} is $\Theta(|E(\texttt{B})|+|V(\texttt{B})|\log|V(\texttt{B})|)+\mathcal{O}(|\mathcal{H}_T|^2)$.

In Proposition~\ref{prop: single}, we show that when the $(L,U)$ suitable pair in Algorithm~\ref{algo: main} comes from the differential privacy framework, $U^d (\alpha) = p(d,\alpha)$ where $p(d,\alpha)$ was defined in Definition~\ref{def:UB}. The following example shows how Algorithm~\ref{algo: main} works.

\begin{example}\label{ex: algorithm}
Consider the graph in Fig.~\ref{fig: counter} again and let the boundary hitting set be $\mathcal{H}_T = \{v_1, v_4\}$. Let $\epsilon = \log(2)$, $\delta = 0$ and fix the mechanism on $\mathcal{H}_T $ as $\alpha = \Pr[\mathcal{M}(v_1) = \texttt{blue}]= 1-\Pr[\mathcal{M}(v_1) = \texttt{red}] = 0.3$ and $\beta =\Pr[\mathcal{M}(v_4) = \texttt{blue}]= 1- \Pr[\mathcal{M}(v_4) = \texttt{red}] = 0.1$. Note that $\dist(v_1, v_4) = 3$. 

Algorithm~\ref{algo: main} we first checks that the mechanism can be extended. Let $j = \texttt{blue}$. Algorithm~\ref{algo: main} checks that $0.3\leq U^3(0.1) 
= p(3,0.1) = 0.7$ and $0.1\leq U^3(0.3) = p(3,0.3) = 0.9$. Therefore, the partial mechanism can be extended.\footnote{For $\epsilon = \log(2)$ and $\delta = 0$,~\eqref{eq:tau0} gives $\tau = 2$ for $\alpha = 0.1$ and $\tau = 1$ for $\alpha = 0.3$. Therefore, $p(3,0.1) = 0.7$ and $p(3,0.3) = 0.9$. }

Next, Algorithm~\ref{algo: main} assigns optimal values to $\Pr[\mathcal{M}(v_2) = \texttt{blue}]$, and $\Pr[\mathcal{M}(v_3) = \texttt{blue}]$. Let us first consider $v_2$ and the upper bound that each vertex $v_1$ and $v_4$ impose on $v_2$. From  Definition~\ref{def:UB}, we have $U^{\dist(v_1,v_2)}(0.3) = p(1,0.3) = 0.6$ and $U^{\dist(v_2,v_4)}(0.1) = p(2,0.1) = 0.4$. Therefore, $\Pr[\mathcal{M}(v_2) = \texttt{blue}] = 0.4$. This is remarkable in the sense that $v_1$ which is the closest vertex to $v_2$ in the boundary hitting set is not the one that imposes the tightest upper bound on $\Pr[\mathcal{M}(v_2) = \texttt{blue}]$. Instead, $v_4$ which is farther from $v_2$, has a more restrictive effect on $v_4$ for taking its true value $\texttt{blue}$. Finally, consider $v_3$. We have $U^{\dist(v_1,v_3)}(0.3) = p(2,0.3) = 0.8$ and $U^{\dist(v_3,v_4)}(0.1) = p(1,0.1) = 0.2$. Therefore, $\Pr[\mathcal{M}(v_3) = \texttt{blue}] = 0.2$.

\end{example}

\section{Proofs}\label{sec:proofs}

In this section, we prove all our results. We start by showing some intermediate lemmas that we use.

\subsection{Intermediate Lemmas}

This lemma states a graph-theoretic result that we use in various proofs.
\begin{lemma}\label{lem: dist}
     Let $\mathcal{G}(V,E)$ be a graph and $u\sim v \in V$. Then for every vertex $w\in V$, we have $\dist(u, w) \leq \dist(v,w)+1$.
\end{lemma}

\begin{proof}

Consider a shortest path $w=w_0,w_1,\ldots,w_{\dist(w,v)}=v$ connecting $w$ to $v$. Since $u$ is a neighbor of $v$, then $w=w_0,w_1,\ldots,w_{\dist(w,v)}=v,u$ is a path connecting $w$ to $u$ with length $\dist(w,v)+1$. Thus, $d(w,u) \leq \dist(w,v)+1$.
\end{proof}

The triangular inequality follows as an extension, where for any $u,v,w \in V$, $\dist(u, w) \leq \dist(u,v)+\dist(v,w)$.

In the next two lemmas, we provide explicit forms for the upper bound function $U_{\texttt{DP}}$ and the lower bound function $L_{\texttt{DP}}$ in the suitable pair, corresponding to the $(\varepsilon,\delta)$-DP mechanism.

\begin{lemma}\label{lem: iff}
The upper bound function in \eqref{eq: upper dp} can be rewritten as follows.
\begin{equation*}
U_{\texttt{DP}}(\alpha) =
\left\{ \begin{aligned} 
  &e^{\varepsilon} \alpha + \delta &{\text{ if  }} & \alpha \leq \frac{1 - \delta}{ e^{\varepsilon}+ 1}, \label{eq: iff:1}  \\
  &\frac{e^{\varepsilon}+\delta-1 +\alpha}{e^{\varepsilon}} &{\text{ if  }} & \frac{1 - \delta}{e^{\varepsilon}+ 1} \leq \alpha \leq 1- \delta,  \\
  & 1 &\text{ if  } & 1- \delta \leq \alpha.
\end{aligned} \right.
\end{equation*}
\end{lemma}

\begin{proof}

In the first case,

\begin{align*}
    \alpha  \leq  \frac{1 - \delta}{e^{\varepsilon}+1} &\Leftrightarrow  (e^{2\varepsilon}-1) \alpha  \leq  (e^{\varepsilon} -1)  (1 - \delta) \\
    &\Leftrightarrow  e^{2\varepsilon} \alpha -\alpha  \leq  e^{\varepsilon} -1  + \delta - e^{\varepsilon} \delta \\ 
    &\Leftrightarrow  e^{2\varepsilon} \alpha   + e^{\varepsilon} \delta \leq \alpha -1 + e^{\varepsilon}   + \delta \\
    &\Leftrightarrow  e^{\varepsilon} \alpha  + \delta \leq  \frac{e^{\varepsilon}+\delta-1+\alpha}{e^{\varepsilon}}, 
\end{align*}
where the first equivalence holds if $e^{\varepsilon} -1 > 0$; if $e^{\varepsilon} -1 = 0$ the last inequality also holds. Also,
\begin{align*}
    \alpha  \leq  \frac{1 - \delta}{e^{\varepsilon}+1} &\Leftrightarrow   \alpha e^{\varepsilon} + \alpha \leq 1-\delta \\ 
    &\Leftrightarrow \alpha e^{\varepsilon}+\delta \leq 1-\alpha\\
    &\Rightarrow \alpha e^{\varepsilon}+\delta \leq  1.
\end{align*}

In the second case,
\begin{align*}
    \alpha  \leq  1 - \delta &\Leftrightarrow \alpha -1 + e^{\varepsilon} + \delta \leq e^{\varepsilon} \\
    &\Leftrightarrow \frac{ e^{\varepsilon}+\delta-1+\alpha}{e^{\varepsilon}} \leq 1.
\end{align*}

Finally,
\begin{align*}
      1 - \delta \leq \alpha &\Leftrightarrow e^\varepsilon \leq e^\varepsilon + \delta -1 + \alpha \\
      &\Leftrightarrow 1 \leq \frac{e^\varepsilon + \delta -1 + \alpha}{e^\varepsilon}.
\end{align*}
\end{proof}

\begin{lemma}\label{lem: lower:iff}
The lower bound function in \eqref{eq: lower dp} can be rewritten as follows.
\begin{equation*}
L_{\texttt{DP}}(\alpha) =
\left\{ \begin{aligned}
&0 & {\text{ if  }}&\alpha \leq \delta, \\
  &\frac{\alpha- \delta}{e^{\varepsilon}} &{\text{ if  }} & \delta \leq \alpha \leq \frac{\delta+e^{\varepsilon}}{ e^{\varepsilon}+ 1},   \\
  &e^{\varepsilon} \alpha- \delta - e^{\varepsilon} +1 &{\text{ if  }} & \alpha \geq \frac{\delta+e^{\varepsilon}}{ e^{\varepsilon}+ 1}.
\end{aligned} \right.
\end{equation*}
\end{lemma}

\begin{proof}
Note that $L_{\texttt{DP}}(\alpha) = 0$ if and only if $e^{\varepsilon} \alpha- \delta - e^{\varepsilon} +1 \leq 0$ and $\frac{\alpha- \delta}{e^{\varepsilon}} \leq 0$. This will happen if and only if $\alpha \leq \frac{e^{\varepsilon}-1+\delta}{e^{\varepsilon}}$ and $\alpha \leq \delta$, respectively. Since $\delta \leq \frac{e^{\varepsilon}-1+\delta}{e^{\varepsilon}}$, the first case follows.

For the second and third cases, 
\begin{align*}
    \alpha  \leq  \frac{\delta+e^{\varepsilon}}{e^{\varepsilon}+1} &\Leftrightarrow  (e^{\varepsilon}+1) \alpha  \leq   ( \delta+e^{\varepsilon}) \\
    &\Leftrightarrow  (e^{2\varepsilon}-1) \alpha   \leq  \delta(e^{\varepsilon}-1)+e^{\varepsilon}(e^{\varepsilon}-1)\\ 
    &\Leftrightarrow e^{2\varepsilon}\alpha-e^{\varepsilon}\delta-e^{2\varepsilon}+e^{\varepsilon} \leq \alpha-\delta\\ 
    &\Leftrightarrow 
    e^{\varepsilon} \alpha- \delta - e^{\varepsilon} +1 \leq \frac{\alpha- \delta}{e^{\varepsilon}}.
\end{align*}
\end{proof}

The next lemma establishes the symmetric nature of suitable pair functions.

\begin{lemma}{\label{lem:symmetry}}
Let $\alpha_1, \alpha_2 \in [0,1]$ and $(L, U)$ be a suitable pair. Then, $\alpha_2 \in [L(\alpha_1),U(\alpha_1)]$ if and only if $\alpha_1 \in [L(\alpha_2),U(\alpha_2)]$.
\end{lemma}
\begin{proof}
Let $(L, U)$ be a suitable pair and suppose that $L(\alpha_1) \leq \alpha_2 \leq U(\alpha_1)$.
In particular, we have $L(\alpha_1) \leq \alpha_2$. Since by definition, $U$ is an increasing function, it follows that $U(L(\alpha_1)) \leq U(\alpha_2)$. On the other hand, according to the definition of $(L, U)$ suitable pair,  $\alpha_1 \leq U(L(\alpha_1))$. Therefore, $\alpha_1  \leq U(L(\alpha_1)) \leq U(\alpha_2)$.

Similarly, we have $\alpha_2 \leq U(\alpha_1)$. Since $L$ is also an increasing function, we have $L(\alpha_2) \leq L(U(\alpha_1))$. According to the definition of $(L, U)$ suitable pair,  $L(U(\alpha_1))\leq \alpha_1$. Therefore, $L(\alpha_2)\leq L(U(\alpha_1))\leq \alpha_1$. Thus, $\alpha_1 \in [L(\alpha_2),U(\alpha_2)]$. 

The reverse implication follows analogously.
\end{proof}

\begin{lemma}\label{lem: generalisedLU}
Mechanism $\mathcal{M}:V \to Q$ is $(L, U)$-private if and only if for every color $j \in Q$ and every two datasets $u,v \in V$ we have:
\[\beta \in [L^d(\alpha),U^d(\alpha)],\]
where $\alpha:=\Pr[\mathcal{M}(u) = j]$, $\beta:=\Pr[\mathcal{M}(v) = j]$ and $d = \dist(u,v)$ is the distance between them.
\end{lemma}
\begin{proof}
We first prove the forward direction through induction. For the case $d = 1$, the claim follows directly from Definition \ref{def:lu:private}. Suppose that the claim holds for $d$. Let $u,v \in V$ such that $\dist(u,v) = d+1$. Consider a path of length $d+1$ from $u$ to $v$. Let $w$ be the one-before-the-last vertex in the path, where $\dist(u,w) = d$. Denote $\gamma:=\Pr[\mathcal{M}(w) = j]$. By the induction hypothesis for the vertices $u, w$ we have $\gamma \in [L^d(\alpha),U^d(\alpha)]$.
Also, since $L, U$ are increasing functions, we have
\[\gamma \geq L^{d}(\alpha) \Rightarrow L(\gamma) \geq L^{d+1}(\alpha)\]
\[\gamma \leq U^{d}(\alpha) \Rightarrow U(\gamma) \leq U^{d+1}(\alpha)\]
Thus, $ L^{d+1}(\alpha) \leq L(\gamma) \leq U(\gamma) \leq U^{d+1}(\alpha)$.
Since, $w$ and $v$ are adjacent, we have $\beta \in [ L(\gamma), U(\gamma)]$. Therefore,
\[
L^{d+1}(\alpha) \leq L(\gamma)\leq \beta \leq U(\gamma) \leq U^{d+1}(\alpha).
\]
Also, the backward can be proved by only considering neighbor vertices.
This completes the proof.
\end{proof}

The next lemma extends the third condition of the suitable pair in Definition~\ref{def: suitable pair}.

\begin{lemma}\label{lem: cond3Gen}
Let $d$ be a positive integer and $(L, U)$ be a suitable pair. Then for every $\alpha \in [0,1]$ we have: 
   \begin{align*}
             U^{d}(\alpha) &\leq 1 - L^{d}(1-\alpha).
        \end{align*} 
\end{lemma}
    \begin{proof}
We prove this by induction on $d$ and the proof only uses the monotonicity of the upper bound function $U$. For a fixed $\alpha \in [0,1]$ and $d = 1$, we have the following equations directly from the third condition in Definition~\ref{def: suitable pair}:
        \begin{align*}
            U(\alpha)+L(1-\alpha)\leq 1 \iff U(\alpha)\leq 1 - L(1-\alpha).
        \end{align*}
Assume that the condition in the lemma is satisfied for $d$. Then, we have:
        \begin{align*}
             U^{d}(\alpha) &\leq 1 - L^{d}(1-\alpha),\\
             \Rightarrow U\big( U^{d}(\alpha) \big) &\leq U\big( 1 - L^{d}(1-\alpha) \big)\\
             &\leq 1 - L\Big(1- \big( 1 - L^{d}(1-\alpha) \big) )\Big)\\
             &= 1 - L\big( L^{d}(1-\alpha)\big)\\
             &= 1 - L^{d+1}(1-\alpha),
        \end{align*}
which completes the proof.
\end{proof}

The following Lemma is a simple extension of the second property of an $(L,U)$ suitable pair, which we will use.
\begin{lemma}\label{lem:compostionbound}
For every $\alpha\in [0,1]$ the following holds:
\[\alpha \leq U^d(L^d(\alpha)). \]
\end{lemma}
\begin{proof}
We prove this by iteratively applying  the definition of $(L,U)$ suitable pair as follows:
\begin{align*}
U^d(L^d(\alpha)) &= U^{d-1} \Big(U \big( L(L^{d-1}(\alpha) \big) \Big)\geq \ U^{d-1}(L^{d-1}(\alpha))\\&\geq
\ldots \geq U(L(\alpha)) \geq \alpha.\end{align*}
\end{proof}
\begin{lemma}\label{lem:UBsufficient}
Let $\alpha , \beta \in [0,1]$ and $d\geq 1$ be given and assume that the following two relations hold: $\alpha \leq U^d(\beta)$ and $\beta \leq U^d(\alpha)$. Then, we have:
\begin{align}
1-\beta &\geq L^d(1-\alpha),\\
1-\alpha &\leq U^d(1-\beta),\\
1-\alpha &\geq L^d(1-\alpha),\\
1-\beta &\leq U^d(1-\alpha),\\
\alpha &\leq L^d(\beta),\\
\beta &\leq L^d(\alpha).
\end{align}
\end{lemma}
\begin{proof}
\begin{align}\label{eq:betarelation1}
\beta \leq U^d(\alpha) \Rightarrow
1-\beta \geq 1-U^d(\alpha) \underset{(a)}{\geq}
L^d(1-\alpha),
\end{align}
where $(a)$ follows from Lemma \ref{lem: cond3Gen}. Similarly, we have $1-\alpha \geq 
L^d(1-\beta)$. 

Now, we apply $U^d$ to the derived relation in \eqref{eq:betarelation1} to write
\begin{align}\label{eq:betarelation2}
U^d(1-\beta) \geq U^d(L^d(1-\alpha)) \underset{(a)}{\geq} 1-\alpha,
\end{align}
where $(a)$ follows from Lemma \ref{lem:compostionbound}.
Similarly, we have $U^d(1-\alpha)\geq 
1-\beta$. 

The one before last inequality follows from~\eqref{eq:betarelation2} and the fact that $U^d(1-\beta) \leq 1-L^d(\beta)$ according to 
Lemma \ref{lem: cond3Gen}. Similarly, we have $L^d(\alpha) \leq 
\beta$. 
\end{proof}
The above lemma gives a sufficient condition for checking the feasibility of mechanism extension only based on the upper bound function $U$ (instead of both $L$ and $U$) and for one mechanism value, say $j \in Q$ (instead of both values in $Q$).

\subsection{Proof of Lemma~\ref{lem: cons}}
We first prove the forward direction. Assume that there exists an $(L, U)$-private mechanism $\mathcal{M}:V \to Q$, extending $\mathcal{M}'$. By the definition of mechanism extension $\mathcal{M}|_{\mathcal{H}_T} = \mathcal{M'}$. Therefore, for any $u, v \in \mathcal{H}_T$, $\Pr[\mathcal{M}(u)= j] = \Pr[\mathcal{M'}(u)= j]$ and $\Pr[\mathcal{M}(v)= j] = \Pr[\mathcal{M'}(v)= j]$. The result then follows directly from Lemma~\ref{lem: generalisedLU}.

Now we prove the reverse direction by an explicit construction of $\mathcal{M}$. Fix $j \in Q$. Assume that for every $u,v \in \mathcal{H}_T$ and $d = \dist(u,v)$, we have $\alpha \leq U^d(\beta)$ and $\beta \leq U^d(\alpha)$, where $\alpha = \Pr[\mathcal{M'}(u)= j]$ and $\beta = \Pr[\mathcal{M'}(v)= j]$. We show how to construct a valid $(L, U)$-private extension $\mathcal{M}$ from $\mathcal{M'}$. First, set $\mathcal{M}|_{\mathcal{H}_T} = \mathcal{M'}$. If $\mathcal{H}_T = V$, we are done. Otherwise, for every $w \notin \mathcal{H}_T$
assign 
\begin{align}\label{eq:construction3}
 \Pr[\mathcal{M}(w)= T(w)] =  \min \limits_{u \in \mathcal{H}_T} U^{\dist(w,u)}(\Pr[\mathcal{M}'(u)=T(w)]).
\end{align}
 For $w\notin \mathcal{H}_T$, let $u_w \in \mathcal{H}_T$ be a minimizer of \eqref{eq:construction3} and let $\alpha_{u_w} := \Pr[\mathcal{M'}(u_w)= T(w)]$. That is, $U^{\dist(w,u_w)}(\alpha_{u_w}) \leq U^{\dist(w,u)}(\Pr[\mathcal{M}'(u)=T(w)])$ for any $u \in \mathcal{H}_T$. We must prove that this mechanism is $(L,U)$-private according to Defintion \ref{def:lu:private}. However, it suffices to fix a $j \in Q$ a priori and for every $v\sim w \in V$ only prove $\alpha_w \leq U(\alpha_v)$ and $\alpha_v \leq U(\alpha_w)$, where $\alpha_v = \Pr[\mathcal{M}(v)= j]$ and $\alpha_w = \Pr[\mathcal{M}(w)= j]$. The reminaing $(L,U)$ relations follow from Lemmas~\ref{lem: generalisedLU} and~\ref{lem:UBsufficient}. We need to consider three cases. 

\textbf{Case 1 ($v,w\in \mathcal{H}_T$):}
This is automatically satisfied by the assumption in the reverse direction of the Lemma (for $d = \dist(v,w) = 1$). 

\textbf{Case 2 ($w \notin \mathcal{H}_T$, $v\in \mathcal{H}_T$):} We consider two subcases depending on the true value of dataset $w$.

\textbf{Subcase 2.1 ($j = T(w)$):} We have
\begin{align*}
\alpha_w &= \Pr[\mathcal{M}(w)= T(w)] = U^{\dist(w,u_w)}(\alpha_{u_w}) \\\nonumber&\underset{(a)}{\leq} U^{\dist(w,v)}(\Pr[\mathcal{M}'(v)=T(w)])\\\nonumber& \underset{(b)}{=} U(\Pr[\mathcal{M}(v)=T(w)]) := U(\alpha_v),
\end{align*}
where $(a)$ follows because $v\in \mathcal{H}_T$ and hence is part of the mininimization in~\eqref{eq:construction3}. In $(b)$ we used the fact that $v,w$ are neighbours and $\mathcal{M}|_{\mathcal{H}_T} = \mathcal{M'}$.
Furthermore,
\begin{align*}
\alpha_v &= \Pr[\mathcal{M}(v)=T(w)] \underset{(a)}{\leq} U^{\dist(v,u_w)}(\alpha_{u_w}) \\&\underset{(b)}{\leq} U^{\dist(v,w)}(U^{\dist(w,u_w)}(\alpha_{u_w})):=U(
\alpha_w),
\end{align*}
where $(a)$ follows the assumption of reverse part of lemma for $v,u_w\in \mathcal{H}_T$. Inequality $(b)$ follows from triangular inequality and the propoerty of $(L,U)$ pair that for $d_1 \geq d_2$ and any $\alpha \in [0,1]$, we have $U^{d_1}(\alpha) \geq U^{d_2}(\alpha)$. The last equality follows from construction~\eqref{eq:construction3} and the fact that $v\sim w$.

\textbf{Subcase 2.2 ($j \neq T(w)$)}: This subcase can be proved by using the results we just proved for $j = T(w)$ and invoking Lemma~\ref{lem:UBsufficient}. Since we now have $1- \alpha_w \leq U(1- \alpha_v)$ and $ 1- \alpha_v \leq U(1- \alpha_w)$.

\textbf{Case 3 ($v,w \notin \mathcal{H}_T$):} Note that we must have $T(v) = T(w)$. Otherwise, $T(v) \neq T(w)$ means that two neighbors $v\sim w$ with different true values are both outside of $\mathcal{H}_T$, which contradicts the definition of a boundary-hitting set. We consider two subcases depending on the true value of these vertices. 

\textbf{Subcase 3.1:} ($j = T(w) = T(v)$) We have
\begin{align}\label{eq:construction4}\nonumber
\alpha_w &:= U^{\dist(w,u_w)}(\alpha_{u_w}) \\&\underset{(a)}{\leq} U^{\dist(w,u_v)}(\alpha_{u_v}) \underset{(b)}{\leq} U^{\dist(w,v)+\dist(v,u_v)}(\alpha_{u_v})\\\nonumber&= U^{\dist(w,v)}(U^{\dist(v,u_v)}(\alpha_{u_v})) := U(\alpha_v),
\end{align}
where $(a)$ follows  because $j = T(v) = T(w)$ and hence, according to construction~\eqref{eq:construction3} there must exist $u_v \in \mathcal{H}_T$ such that $\alpha_v := U^{\dist(v,u_v)}(\alpha_{u_v})$. Inequality $(b)$ follows from applying the triangular inequality.

By swapping the roles of $\alpha_v$ and $\alpha_w$ in the above, we obtain $\alpha_v \leq U(\alpha_w)$.

\textbf{Subcase 3.2 ($j \neq T(w) = T(v)$):} This follows directly from case 3.1 just proved and Lemma~\ref{lem:UBsufficient}.

\subsection{Proof of Theorem~\ref{theo: algo 1}}

We note that the value assigned to $\Pr[\mathcal{M}(w)=T(w)]$ in Algorithm~\ref{algo: main} is the same as in \eqref{eq:construction3}. Thus, by Lemma~\ref{lem: cons} the algorithm finds an $(L,U)$-private extension if it exists.

We now show that the private mechanism is the unique optimal mechanism.

Suppose there exists another $(L,U)$-private mechanism $\mathcal{M}_2$ which extends $\mathcal{M}'$. Let $w \in V$ be a vertex with $T(w)=j$. Then, by Lemma~\ref{lem: generalisedLU}, for every $u \in \mathcal{H}_T$, it follows that $\Pr[\mathcal{M}_2(w)=j] \leq U^{\dist (w,u)} \left( \Pr[\mathcal{M}'(u)=j] \right)$. Thus, $\Pr[\mathcal{M}_2(w)=j] \leq \min \limits_{u \in \mathcal{H}_T} U^{\dist(w,u)}(\Pr[\mathcal{M}'(u)=j]) = \Pr[\mathcal{M}(w) = j]$. Thus, $\Pr[\mathcal{M}_2(w)=T(w)] \leq \Pr[\mathcal{M}(w) = T(w)]$ for every $w \in V$. Thus, $\mathcal{M}$ dominates every other $\mathcal{M}_2$ and is therefore the unique optimal.

\subsection{Proof of Proposition~\ref{prop: uniqueness}}

The functions $U_{\texttt{DP}}$ and $L_{\texttt{DP}}$ are a suitable pair if they are increasing functions and, for every $\alpha \in [0,1]$ they satisfy the following conditions:
\begin{enumerate}
    \item $L_{\texttt{DP}}(\alpha) \leq \alpha \leq U_{\texttt{DP}}(\alpha)$,
    \item $L_{\texttt{DP}}(U_{\texttt{DP}}(\alpha)) \leq \alpha \leq U_{\texttt{DP}}(L_{\texttt{DP}}(\alpha))$,
    \item $U_{\texttt{DP}}(\alpha) \leq 1- L_{\texttt{DP}}(1-\alpha)$.
\end{enumerate}
We prove each condition as its own Lemma. We first prove the third condition as it will be used in proving the first condition.

\begin{lemma}\label{lem: u+l=1} $U_{\texttt{DP}}(\alpha) = 1- L_{\texttt{DP}}(1 - \alpha)$. 
\end{lemma}
\begin{proof}
Assume $L_{\texttt{DP}}(1 - \alpha) > 0$ and $U_{\texttt{DP}}(\alpha) < 1$. Then,
\begin{align} \label{eq: upper dp lem u+l}
  U_{\texttt{DP}}(\alpha)  := \min(e^{\varepsilon} \alpha + \delta , \frac{e^{\varepsilon}+\delta-1 +\alpha}{e^{\varepsilon}}). 
\end{align}
And,
\begin{align} \label{eq: lower dp lem u+l}
 L_{\texttt{DP}}(1 - \alpha) := \max (1 - e^{\varepsilon} \alpha- \delta , \frac{1 - \alpha- \delta}{e^{\varepsilon}}).
\end{align}
Therefore, we have:
\begin{align*} \label{eq: u+l}
 U_{\texttt{DP}}(\alpha) + L_{\texttt{DP}}(1- \alpha) &= \min(e^{\varepsilon} \alpha + \delta , \frac{e^{\varepsilon}+\delta-1 +\alpha}{e^{\varepsilon}}) \\
 &\quad +
 \max (1 - e^{\varepsilon} \alpha - \delta , \frac{1 - \alpha- \delta}{e^{\varepsilon}}).\nonumber
\end{align*}
Then, we have:
\begin{align*}
  U_{\texttt{DP}}(\alpha) = e^{\varepsilon} \alpha + \delta &\iff \\
  e^{\varepsilon} \alpha + \delta \leq \frac{e^{\varepsilon}+\delta-1 +\alpha}{e^{\varepsilon}} &\iff \\
  e^{\varepsilon} \alpha + \delta \leq 1+ \frac{\delta-1 +\alpha}{e^{\varepsilon}} &\iff \\
   \frac{1- \delta -\alpha}{e^{\varepsilon}}\leq 1 - e^{\varepsilon} \alpha - \delta  &\iff \\
  L_{\texttt{DP}}(1 - \alpha) = 1 - e^{\varepsilon} \alpha - \delta.
\end{align*}
Therefore, $U_{\texttt{DP}}(\alpha) + L_{\texttt{DP}}(1- \alpha) = 1$ whenever $U_{\texttt{DP}}(\alpha) = e^{\varepsilon} \alpha + \delta$. Similarly, 
\begin{align*}
  U_{\texttt{DP}}(\alpha) = \frac{e^{\varepsilon}+\delta-1 +\alpha}{e^{\varepsilon}} &\iff \\
  e^{\varepsilon} \alpha + \delta \geq \frac{e^{\varepsilon}+\delta-1 +\alpha}{e^{\varepsilon}} &\iff \\
  e^{\varepsilon} \alpha + \delta \geq 1+ \frac{\delta-1 +\alpha}{e^{\varepsilon}} &\iff \\
   \frac{1- \delta -\alpha}{e^{\varepsilon}}\geq 1 - e^{\varepsilon} \alpha - \delta  &\iff \\
  L_{\texttt{DP}}(1 - \alpha) = \frac{1- \delta -\alpha}{e^{\varepsilon}}.
\end{align*}
Therefore, $U_{\texttt{DP}}(\alpha) + L_{\texttt{DP}}(1- \alpha) = 1$ again.
Finally, Lemma \ref{lem: iff} implies that:
 \begin{align}
 U_{\texttt{DP}}(\alpha) = 1 \iff  1 - \delta \leq \alpha \iff\\
1 - \delta - \alpha \leq 0 \iff\frac{1 - \delta - \alpha}{e^{\varepsilon}} \leq 0,
 \end{align}
and
\begin{align}
U_{\texttt{DP}}(\alpha) = 1 \iff  1 - \delta \leq \alpha \iff\\
1 - \delta - \alpha \leq 0 \Rightarrow 1 - \delta - e^{\varepsilon}\alpha  \leq 0.
 \end{align}
 Therefore, $L_{\texttt{DP}}(1 - \alpha)  = 0$ and $U_{\texttt{DP}}(\alpha) + L_{\texttt{DP}}(1- \alpha) = 1.$
\end{proof}

\begin{lemma}
    $L_{\texttt{DP}}(\alpha) \leq \alpha \leq U_{\texttt{DP}}(\alpha)$.
\end{lemma}

\begin{proof}

We first prove $U_{\texttt{DP}} (\alpha) \geq \alpha$. We have 3 cases based on the value of $U_{\texttt{DP}}(\alpha)$ following from Lemma \ref{lem: iff}.

\textbf{Case 1} (when $U_{\texttt{DP}} (\alpha) = e^{\varepsilon} \alpha + \delta$):
Since $e^{\varepsilon} \geq 1$ and $\delta \geq 0$ it holds that $U_{\texttt{DP}}(\alpha) =e^{\varepsilon} \alpha + \delta \geq \alpha$.

\textbf{Case 2} (when $U_{\texttt{DP}} (\alpha) = \frac{e^{\varepsilon}+\delta-1 +\alpha}{e^{\varepsilon}}$): Since $e^{\varepsilon} \geq 1$, $\alpha \leq 1$ and $\delta \geq 0$ we have
\begin{align*}
    1 +\frac{\delta}{(e^{\varepsilon}-1)} \geq 1 \geq \alpha &\iff \\
    1 +\frac{\delta}{(e^{\varepsilon}-1)} \geq \alpha &\iff \\
    (e^{\varepsilon}-1) +\delta \geq (e^{\varepsilon}-1) \alpha &\iff \\
    e^{\varepsilon}+\delta-1 \geq e^{\varepsilon} \alpha - \alpha &\iff \\
    e^{\varepsilon}+\delta-1 +\alpha \geq e^{\varepsilon} \alpha &\iff \\
    \frac{e^{\varepsilon}+\delta-1 +\alpha}{e^{\varepsilon}} \geq \alpha.
\end{align*}

\textbf{Case 3} (when $U_{\texttt{DP}} (\alpha) = 1$): Follows from $\alpha \leq 1$.

We now show that $L_{\texttt{DP}}(\alpha) \leq \alpha$. It follows from Lemma~\ref{lem: u+l=1} and what we just proved that $1 - L_{\texttt{DP}}(\alpha) = U_{\texttt{DP}} (1 - \alpha)  \geq 1 - \alpha$. Therefore, $\alpha \geq L_{\texttt{DP}}(\alpha)$.

\end{proof}

\begin{lemma}\label{lem: UL=LU}
$L_{\texttt{DP}}(U_{\texttt{DP}}(\alpha)) \leq \alpha \leq U_{\texttt{DP}}(L_{\texttt{DP}}(\alpha))$.
\end{lemma}
\begin{proof}
First, we prove the first inequality. We write $L_{\texttt{DP}}(U_{\texttt{DP}}(\alpha))$ as follows:
\begin{align} \label{eq: lower dp lem special cases}
 L_{\texttt{DP}}(U_{\texttt{DP}}(\alpha)) = \max (0, e^{\varepsilon} U_{\texttt{DP}}(\alpha)- \delta - e^{\varepsilon} +1 , \frac{U_{\texttt{DP}}(\alpha)- \delta}{e^{\varepsilon}}).
\end{align}
We will have $L_{\texttt{DP}}(U_{\texttt{DP}}(\alpha)) \leq \alpha$ if and only if each expression in the $\max$ expression above is less than $\alpha$. Obviously $0 \leq \alpha$ is true. The other conditions are equivalently written below:
\begin{align*}
 e^{\varepsilon} U_{\texttt{DP}}(\alpha)- \delta - e^{\varepsilon} +1  \leq \alpha &\iff U_{\texttt{DP}}(\alpha) \leq \frac{ \alpha + \delta + e^{\varepsilon} -1}{e^{\varepsilon}},\\
 \frac{U_{\texttt{DP}}(\alpha)- \delta}{e^{\varepsilon}} \leq \alpha &\iff 
 U_{\texttt{DP}}(\alpha) \leq e^{\varepsilon} \alpha +\delta,
\end{align*}
where the inequalities are true according to the definion of $U_{\texttt{DP}}$ in~\eqref{eq: upper dp}. Similarly, we have $\alpha \leq U_{\texttt{DP}}(L_{\texttt{DP}}(\alpha))$. 
\end{proof}

The last step of completing the proof of Proposition~\ref{prop: uniqueness} is to prove that both functions $U_{\texttt{DP}}$ and $L_{\texttt{DP}}$ are increasing. However, this can be done by using the equivalent definition of $U_{\texttt{DP}}(\alpha)$ in Lemma \ref{lem: iff}. In the first two cases, this function is linear and in the third case, it is a constant. Also, this function is continuous. For proving the same for $L_{\texttt{DP}}(\alpha)$, we use the fact that  $U_{\texttt{DP}}(1 - \alpha ) = 1 - L_{\texttt{DP}}(\alpha)$. Therefore $L_{\texttt{DP}}(\alpha) = 1 - U_{\texttt{DP}}(1 - \alpha)$ is an increasing function.

\subsection{Proof of Proposition~\ref{prop: single}}
Via Proposition~\ref{prop: uniqueness}, $(L_{\texttt{DP}},U_{\texttt{DP}})$ are $(L,U)$ suitable pairs. Therefore, by Lemma~\ref{lem: generalisedLU},  $\mathcal{M}$ will be $(\varepsilon,\delta)$-DP if an only if for every color $j \in Q$ and every $u,v \in V$ we have:
\[\beta \in [L^d(\alpha),U^d(\alpha)],\]
where $\alpha:=\Pr[\mathcal{M}(u) = j]$, $\beta:=\Pr[\mathcal{M}(v) = j]$ and $d = \dist(u,v)$ is the distance between them. Therefore, we first compute and simplify $U_{\texttt{DP}}^d(\alpha)$. We use Lemma~\ref{lem: iff} for this.

\textbf{Part 1} (Showing that $\beta \leq U_{\texttt{DP}}^d(\alpha)$): First assume that there exists some $\tau$ such that for all $1 \leq  d  \leq \tau$, we have $\min (  e^{\varepsilon} U_{\texttt{DP}}^{d-1}(\alpha) + \delta, \frac{ U_{\texttt{DP}}^{d-1}(\alpha) -1 + e^{\varepsilon} + \delta}{e^{\varepsilon}},1) = e^{\varepsilon} U_{\texttt{DP}}^{d-1}(\alpha) + \delta$. That is, the first case in Lemma~\ref{eq: iff:1} is the tightest upper bound on $U_{\texttt{DP}}^d(\alpha) = U(U_{\texttt{DP}}^{d-1}(\alpha))$. We will soon find the largest $\tau$ for which this can happen. Iterating over $d = \tau, \tau-1, \cdots, 1$, we will calculate the closed-form expression through induction as follows.
\begin{align}
    \mathllap U_{\texttt{DP}}^{d}(\alpha) & =  e^{\varepsilon} U_{\texttt{DP}}^{d-1}(\alpha) + \delta \nonumber\\
      & = e^{\varepsilon} \big( e^{\varepsilon} U_{\texttt{DP}}^{d-2}(\alpha) + \delta \big) + \delta  \nonumber\\
      & \cdots \nonumber\\
      & = e^{\varepsilon} \Bigg( e^{\varepsilon} \Big( \cdots \big( e^{\varepsilon} U_{\texttt{DP}}(\alpha) + \delta \big) \cdots + \delta \Big) + \delta \Bigg) + \delta \nonumber\\ 
      & = e^{\varepsilon} \Bigg( e^{\varepsilon} \Big( \cdots \big( e^{\varepsilon} ( e^{\varepsilon} \alpha + \delta ) + \delta \big) \cdots + \delta \Big) + \delta \Bigg) + \delta \nonumber \\
      & = e^{d \varepsilon} \alpha + \delta \frac{e^{d \varepsilon}-1}{e^{\varepsilon}-1}.\label{eq:lastn}
\end{align}
We want to find the last index for which the iterations \eqref{eq:lastn} hold. First, on the one hand, $\tau$ satisfies
\begin{align}\label{eq:tau1}
U_{\texttt{DP}}^{\tau}(\alpha) &  = e^{\tau \varepsilon} \alpha + \delta \frac{e^{ \tau \varepsilon}-1}{e^{\varepsilon}-1}.
\end{align}
On the other hand, by the definition of $\tau$, for $d = \tau + 1$, the second case in Lemma~\ref{eq: iff:1} will give the tightest upper bound on $U_{\texttt{DP}}^{\tau+1}(\alpha)$. That is, 
\begin{align*}
&\min (  e^{\varepsilon} U_{\texttt{DP}}^{\tau}(\alpha) + \delta, \frac{ U_{\texttt{DP}}^{\tau}(\alpha) -1 + e^{\varepsilon} + \delta}{e^{\varepsilon}},1) \\&= \frac{ U_{\texttt{DP}}^{\tau}(\alpha) -1 + e^{\varepsilon_\tau} + \delta}{e^{\varepsilon}}.
\end{align*}
Therefore, from Lemma~\ref{lem: iff}, we must have
\begin{align}\label{eq:tau2}
U_{\texttt{DP}}^{\tau}(\alpha) \geq \frac{1 -\delta}{ e^{\varepsilon} + 1}. 
\end{align}
Combining \eqref{eq:tau2} and  \eqref{eq:tau1} gives the $\tau$ as defined in \eqref{eq:tau0}.

Now, note that the terms in the three conditions of Lemma~\ref{lem: iff} are all monotonically non-decreasing. Also, the rate of increase with respect to $\alpha$ for the conditions of Lemma~\ref{lem: iff} is respectively, $e^{\varepsilon}$, $1/e^{\varepsilon}$ and $0$. Therefore, once the second case in Lemma~\ref{lem: iff} becomes the tightest, it will remain so. So is for the third case. In summary, the cases in Lemma~\ref{lem: iff} do not ``toggle'' or ``alternate'' in providing the tightest upper bound for $d > \tau$. 

The last step is to provide a closed-form expression for the iterations $\tau < d$, where the second case of Lemma~\ref{lem: iff} is active. Starting with $d = \tau+1$, we will have
\begin{align}
U_{\texttt{DP}}^{\tau +1}(\alpha) &=  \frac{ U_{\texttt{DP}}^{\tau}(\alpha) -1 + e^{\varepsilon} + \delta}{e^{\varepsilon}} \nonumber \\
& =  \frac{ \big(  
e^{\tau \varepsilon} \alpha + \delta \frac{e^{ \tau \varepsilon}-1}{e^{\varepsilon}-1}
\big) -1 + e^{\varepsilon} + \delta}{e^{\varepsilon}} \nonumber \\ \label{eq:dtau1}&=e^{(\tau-1)\varepsilon}\alpha+1-\frac{1}{e^{\varepsilon}}+\frac{\delta(e^{\tau\varepsilon}+e^{\varepsilon}-2)}{e^{\varepsilon}(e^{\varepsilon}-1)},
\end{align}
which matches the second case in \eqref{eq:UB} for $d=\tau+1$.
One can use~\eqref{eq:dtau1} to continue with $d > \tau + 1$.

\textbf{Part 2} (Showing that $\beta \geq L_{\texttt{DP}}^d(\alpha)$):  From Lemma~\ref{lem: cond3Gen}, we have that $U_{\texttt{DP}}^{d}(1-\alpha) \leq 1 - L_{\texttt{DP}}^{d}(\alpha)$ or $L_{\texttt{DP}}^{d}(\alpha) \leq 1-U_{\texttt{DP}}^{d}(1-\alpha)$. From Lemma~\ref{lem:UBsufficient}, we have that 
$\beta \leq U_{\texttt{DP}}^d(\alpha)$ implies $\beta \geq 1-U_{\texttt{DP}}^d(1-\alpha)$. The result follows immediately. 

\subsection{Proof of Theorem~\ref{theo: uniqueness}}

By Proposition~\ref{prop: uniqueness}, $(L_{\texttt{DP}},U_{\texttt{DP}})$ is a suitable pair. Thus, Algorithm~\ref{algo: main} can be utilized to find the unique optimal $(\epsilon,\delta)$-DP mechanism. Then, according to the proof of Proposition~\ref{prop: single}, the expression for $U_{\texttt{DP}}^d$ is equal to $p(d,\alpha)$ in Definition~\ref{def:UB}.

\section{On the Generality of Reasonable Utility}\label{sec:discussion1}

One of our main contributions is the notion of reasonable utility that we present in Definition \ref{def: reasonable utility} for binary DP mechanisms. To the best of our knowledge,  when restricted to binary mechanisms, all utilities previously suggested in the literature conform to this concept, including those in~\cite{natashalaplace,ghosh2012universally,staircase,quan2016optimal,geng2015optimal,holahan_LP_DP,multiparty,hamming_DP}. Another general notion of utility for binary DP mechanisms is the following.

\begin{definition}{(Strong reasonable utility)}.
Let $T$ be the true function and $\mathfrak{M}_{\epsilon,\delta}$ the family of $(\epsilon,\delta)$-DP mechanisms. Let $\mathcal{U}:\mathfrak{M}_{\epsilon,\delta}\to \mathbb{R}^{\geq 0}$ be a utility function that assigns non-negative real numbers to the mechanisms in $\mathfrak{M}_{\epsilon,\delta}$. We say $\mathcal{U}$ is a strong reasonable utility function if the following conditions hold.
\begin{itemize}
\item[1)] If $\mathcal{M}_1,\mathcal{M}_2 \in \mathfrak{M}_{\epsilon,\delta}$ and for all $v \in V$, $\Pr [\mathcal{M}_1(v)=T(v)] \geq \Pr [\mathcal{M}_2(v)=T(v)] $ then $\mathcal{U}(\mathcal{M}_1) \geq \mathcal{U}(\mathcal{M}_2)$.
\item[2)] If $\mathcal{M}_1,\mathcal{M}_2 \in \mathfrak{M}_{\epsilon,\delta}$ and for all $v \in V$, $\Pr [\mathcal{M}_1(v)=T(v)] \geq \Pr [\mathcal{M}_2(v)=T(v)] $ and for at least one $x$, the inequality is strict, then $\mathcal{U}(\mathcal{M}_1) > \mathcal{U}(\mathcal{M}_2)$.
\end{itemize}
\end{definition}

Similar to the proof of Theorem \ref{theo: algo 1}, one can argue that for any partial mechanism and any strong reasonable utility function $\mathcal{U}$, if an $(\varepsilon,\delta)$ differentially private extension exists, then there exists a unique optimal extension with respect to $\mathcal{U}$. This optimal extension is independent of the actual utility function, and, moreover, Algorithm \ref{algo: main} outputs this unique optimum extension mechanism.

\section{Extending to Non-Binary Mechanisms}\label{sec:discussion2}

It is not clear how to extend the notion of reasonable utility to non-binary DP mechanisms. The difficulty arises because, in binary mechanisms (e.g., $Q=\{\texttt{blue},\texttt{red}\}$), optimizing the output probability for the true value (e.g., $\texttt{blue}$) is straightforward and is equivalent to minimizing its incorrect outcome (e.g., $\texttt{red}$). In contrast, with more options (e.g., $Q=\{\texttt{blue},\texttt{red},\texttt{green}\}$), it is clear that the probability of the true outcome ($\texttt{blue}$) should be optimized, but the treatment of other outcomes ($\texttt{red}$, $\texttt{green}$) lacks clarity without further assumptions. This issue has been explored using a lexicographical ordering \cite{zhou2022rainbow} and a dominance ordering \cite{gu2023generalized} based assumption. Optimal mechanisms for these cases have been proposed for the special boundary homogeneous case.

\bibliographystyle{IEEEtran}
\bibliography{ref}

\end{document}

%% file: tikz/randomGraph.tikz
\begin{tikzpicture}[scale=1]

\draw[color=black, thick] (1,1) -- (2,1) -- (2,2) -- cycle;
\draw[color=black, thick] (2,2) -- (3,1) -- (4,2);
\draw[color=green, thick] (4,2) -- (4,3);
\draw[color=black, thick] (4,3) -- (3,4);
\draw[color=green, thick] (3,4) -- (2,3);

\draw[color=black, thick] (2,3) -- (1,3);
\draw[color=black, thick] (3,4) -- (3,3);
\draw[color=black, thick] (4,3) -- (5,3) -- (6,2) -- (7,1) -- (6,1) -- (6,2) -- (5,1) -- (5,2);
\draw[color=green, thick] (5,2) -- (4,2);

\draw[color=black, thick] (5,3) -- (6,3);
\draw[color=black, thick] (5,3) -- (6,4);
\draw[color=green, thick] (6,4) -- (7,4);
\draw[color=green, thick] (6,4) -- (7,3);
\draw[color=black, thick] (7,2) -- (8,2);
\draw[color=green, thick] (7,1) -- (7,2);

\draw[blue,fill=blue] (1,1) circle (.5ex);
\node at (1,0.7) {{\footnotesize $a$}};
\draw[blue,fill=blue] (2,1) circle (.5ex);
\node at (2,0.7) {{\footnotesize $b$}};
\draw[blue,fill=blue] (3,1) circle (.5ex);
\node at (3,0.7) {{\footnotesize $c$}};
\draw[blue,fill=blue] (2,2) circle (.5ex);
\node at (2.35,2) {{\footnotesize $g$}};
\draw[blue,fill=blue] (4,2) circle (.5ex);
\node at (4.04,1.7) {{\footnotesize $h$}};
\draw[blue,fill=blue] (1,3) circle (.5ex);
\node at (1,3.3) {{\footnotesize $k$}};
\draw[blue,fill=blue] (2,3) circle (.5ex);
\node at (2,3.3) {{\footnotesize $\ell$}};

\draw[red,fill=red] (5,1) circle (.5ex);
\node at (5,0.7) {{\footnotesize $d$}};
\draw[red,fill=red] (6,1) circle (.5ex);
\node at (6,0.7) {{\footnotesize $e$}};
\draw[red,fill=red] (7,1) circle (.5ex);
\node at (7,0.7) {{\footnotesize $f$}};
\draw[red,fill=red] (5,2) circle (.5ex);
\node at (5.25,2) {{\footnotesize $i$}};
\draw[red,fill=red] (6,2) circle (.5ex);
\node at (6.4,2) {{\footnotesize $j$}};
\draw[red,fill=red] (3,3) circle (.5ex);
\node at (3.1,2.7) {{\footnotesize $m$}};
\draw[red,fill=red] (4,3) circle (.5ex);
\node at (4.1,3.3) {{\footnotesize $n$}};
\draw[red,fill=red] (5,3) circle (.5ex);
\node at (4.9,3.3) {{\footnotesize $o$}};
\draw[red,fill=red] (6,3) circle (.5ex);
\node at (6.3,3) {{\footnotesize $p$}};
\draw[red,fill=red] (3,4) circle (.5ex);
\node at (3.3,4.1) {{\footnotesize $q$}};
\draw[red,fill=red] (6,4) circle (.5ex);
\node at (6,4.3) {{\footnotesize $r$}};

\draw[blue,fill=blue] (7,4) circle (.5ex);
\node at (7,4.3) {{\footnotesize $s$}};
\draw[blue,fill=blue] (7,3) circle (.5ex);
\node at (7.3,3) {{\footnotesize $t$}};

\draw[blue,fill=blue] (7,2) circle (.5ex);
\node at (7,2.2) {{\footnotesize $u$}};
\draw[blue,fill=blue] (8,2) circle (.5ex);
\node at (8,2.2) {{\footnotesize $v$}};

\end{tikzpicture}

%% file: tikz/counter.tikz
\begin{tikzpicture}[scale=1.5]
\draw[color=black, thick] (10,2.5) -- (13,2.5);
\node at (10.5,2.7) {{\footnotesize $e_1$}};
\node at (11.5,2.7) {{\footnotesize $e_2$}};
\node at (12.5,2.7) {{\footnotesize $e_3$}};

\draw[red,fill=red] (10,2.5) circle (.5ex);
\node at (10,2.2) {{\footnotesize $v_1$}};
\draw[blue,fill=blue] (11,2.5) circle (.5ex);
\node at (11,2.2) {{\footnotesize $v_2$}};
\draw[blue,fill=blue] (12,2.5) circle (.5ex);
\node at (12,2.2) {{\footnotesize $v_3$}};
\draw[red,fill=red] (13,2.5) circle (.5ex);
\node at (13,2.2) {{\footnotesize $v_4$}};

\end{tikzpicture}